\documentclass[12pt]{article}

\usepackage{amsmath,amssymb,amsthm,amsfonts,natbib}
\usepackage{fullpage}
\usepackage{enumitem}
\usepackage{hyperref}

\newtheorem{definition}{Definition}
\newtheorem{theorem}{Theorem}
\newtheorem{lemma}{Lemma}
\newtheorem{proposition}{Proposition}
\newtheorem{corollary}{Corollary}
\newtheorem{remark}{Remark}

\title{Continuation--Performance Decomposition\\
in Dynamic Games with Irreversible Failure}
\author{
Nicholas H. Kirk$^{1,2}$\thanks{Email: \texttt{kirk@markovian.group}}\\
$^{1}$ Markovian Group\\
$^{2}$ Sa\"id Business School, University of Oxford
}

\date{}

\begin{document}
\maketitle

\begin{abstract}
Once failure is irreversible, continuation payoffs cannot be meaningfully aggregated across strategies that differ in their survival properties. Standard scalar evaluation sidesteps this by arbitrarily completing payoffs beyond termination, but such completions are extrinsic to the game form. This paper introduces continuation–performance decomposition (CPD), proving that any evaluation satisfying natural regularity conditions, such as failure-completion invariance, survival locality, and local expected-utility coherence---must separate continuation from performance lexicographically. Continuation priority thus emerges as a consequence of well-posed evaluation, not as a behavioral assumption. We establish equivalence between CPD and the limit of games with diverging failure penalties, show that viability is a game-form invariant independent of payoffs, and apply the framework to bank runs: preemptive withdrawals reflect rational viability vetoes rather than coordination failure when continuation is distributively asymmetric. CPD resolves a representational problem, not a preference problem.
\end{abstract}

\section{Introduction}

Many dynamic games of economic interest involve irreversible failure:
bankruptcy, default, technological collapse, institutional breakdown.
Once failure occurs, the game terminates permanently.
Standard intertemporal evaluation aggregates outcomes over time
using a single scalar, typically discounted expected utility.
This aggregation implicitly assumes that outcomes that \emph{terminate} the game
and outcomes that merely \emph{reduce continuation payoffs}
are commensurable at a fixed rate.

Under irreversible failure, this assumption is fragile.
Failure destroys the domain over which continuation payoffs are defined.
This paper introduces a change of viewpoint that makes this asymmetry explicit at the level
of representation.
We propose a \emph{continuation--performance decomposition} (CPD) of dynamic games
with absorbing failure.
The decomposition separates the evaluation of continuation itself
from the evaluation of payoffs conditional on continuation.
In the decomposed representation, continuation is prior to performance as a matter of
well-posed evaluation, not as a preference or behavioral assumption.
When failure is intolerable—whether modeled by undefined continuation or by arbitrarily large penalties—scalar expected-utility evaluation collapses: either all strategies with positive failure probability become indistinguishable or performance differences are rendered irrelevant. Continuation–performance decomposition makes explicit the non-scalar structure to which any such evaluation converges in the limit.


\section{Related Work}

This paper is related to several literatures that study decision-making, control,
and strategic interaction under risk, uncertainty, and constraints.
The present contribution differs from these approaches in both object and method:
rather than introducing survival as a constraint or a preference primitive,
we show that irreversible failure creates a problem of well-posed evaluation,
which necessitates a decomposition of continuation and performance.

\subsection{Safety-first and non-Archimedean decision theory}

A long tradition in decision theory recognizes that certain risks,
such as ruin or bankruptcy, may warrant priority over expected returns.
The classical safety-first criterion of \citet{Roy1952} evaluates portfolios
by minimizing the probability of falling below a disaster threshold.
Related approaches impose probabilistic safety guarantees or minimum
achievement constraints \citep{Kataoka1963}.

More formally, several authors have studied non-Archimedean or lexicographic
extensions of expected utility, in which some outcomes are given infinite or
lexicographically prior weight \citep{Hammond1988}.
Other work relaxes standard axioms of expected utility, such as completeness,
to accommodate partial comparability of outcomes \citep{Aumann1962}.

These approaches treat survival priority as a \emph{preference assumption}:
agents are postulated to rank certain events lexicographically or to refuse
trade-offs below a threshold.
By contrast, the present paper does not modify preferences.
Instead, we show that when failure is absorbing, unconditional aggregation
is ill posed unless payoffs are arbitrarily extended beyond termination.
Continuation--performance decomposition derives lexicographic structure
as a consequence of well-posed evaluation, not as a behavioral postulate.

\subsection{Viability theory, constrained control, and safe reinforcement learning}

Viability theory, constrained control, and related approaches address persistence by imposing
\emph{constraints} on admissible trajectories.
Viability theory characterizes the set of states and controls that keep the system within a
prescribed safe set \citep{Aubin1991}.
Constrained Markov decision processes optimize a scalar objective subject to probabilistic or
almost-sure constraints on constraint violations \citep{Altman1999}.
In reinforcement learning and control, closely related formulations appear as reach--avoid,
stochastic reachability, and safe reinforcement learning problems \citep{GarciaFernandez2015}.

These frameworks solve constrained optimization problems on a \emph{fixed payoff domain}:
payoffs are defined on all histories, and survival or safety restricts the feasible set of strategies.
Continuation is therefore encoded either as a hard constraint or as a feasibility requirement
within scalar optimization.

The contribution of the present paper is orthogonal.
We do not introduce new constraints, nor do we argue that these approaches ``miss'' survival.
Instead, we identify a distinct issue that arises when failure is absorbing:
different strategies induce continuation problems defined on \emph{different payoff domains}.
Once failure occurs, the continuation problem itself ceases to exist as part of the game form.
Comparisons across strategies that differ in their continuation properties therefore involve
objects with non-comparable domains of definition.

Continuation--performance decomposition addresses this representational problem.
It does not alter the feasible set or impose additional constraints.
Rather, it provides an intrinsic evaluation of strategy profiles when the payoff domain is
endogenously state-dependent, by separating continuation from performance \emph{prior} to any
optimization conditional on survival.

\subsection{Information, learning, and strategic interaction}

A large literature studies the value of information and the comparison of
information structures.
The classical Blackwell order ranks experiments by their informativeness
and implies monotonicity of value under expected utility \citep{Blackwell1953}.
In economics, this logic underpins results on the private and social value
of information \citep{Hirshleifer1971}, as well as modern work on Bayesian
persuasion and information design \citep{KamenicaGentzkow2011,BergemannMorris2016}.

Relatedly, global games and coordination models show how small informational
perturbations can select equilibria or trigger crises \citep{MorrisShin1998}.
These effects typically operate through equilibrium selection or strategic
complementarities.

Our mechanism is orthogonal.
We do not rely on multiplicity, coordination failure, or equilibrium selection.
Instead, irreversible failure induces local concavity in continuation values
near viability boundaries, so that greater informational precision can reduce
the probability of continuation itself.
This reversal arises because evaluation must be decomposed once failure is
absorbing; the standard monotonicity results for information rely on
unconditional expected utility, which is not intrinsic in such environments.

Relatedly, \cite{Kirk2026Irreversible} studies dynamic games with irreversible failure in which increased informational precision can reduce survival probabilities; the present paper differs in object and method by addressing the well-posedness of strategy evaluation itself rather than the comparative statics of information.

\subsection{Positioning}

Across these literatures, survival is treated either as a preference primitive
(e.g.\ safety-first or lexicographic expected utility) or as a constraint within
scalar optimization (e.g.\ viability theory, constrained MDPs, safe reinforcement learning).
The present contribution is not that survival matters, nor that agents should prioritize it.
Rather, the contribution is to explain \emph{why scalar aggregation itself ceases to be intrinsic}
once failure is absorbing.

When different strategy profiles induce continuation problems defined on different domains,
any unconditional scalar evaluation implicitly relies on an extrinsic completion of payoffs
after termination.
Continuation--performance decomposition resolves this representational problem by separating
continuation from performance and recovering standard optimality only where the payoff domain
is well defined.
In this sense, CPD contributes a representation and decoupling result, not a new preference,
constraint, or behavioral assumption.

\section{Dynamic Games with Irreversible Failure}

\begin{definition}[Game form]
A dynamic game form is a tuple
\[
\mathcal G=(S,\{A_i\}_{i\in I},P,\Omega,F),
\]
where $S$ is the state space, $A_i(s)$ are action sets,
$P(\cdot\mid s,a)$ is the transition kernel,
$\Omega$ is the information structure,
and $F\subset S$ is an absorbing failure set:
\[
s\in F \;\Rightarrow\; P(s\mid s,a)=1 \quad \forall a.
\]
\end{definition}

An instance game augments $\mathcal G$ with payoff functions $u_i$
and an initial state or belief.
The game form is held fixed throughout.

\begin{definition}[Failure time]
For a strategy profile $\sigma$, let $T(\sigma)\in\mathbb N\cup\{\infty\}$
denote the (random) hitting time of $F$.
\end{definition}

\section{Continuation--Performance Decomposition}

Let $\Sigma$ denote the set of admissible strategy profiles.

\paragraph{Domain convention.}
Throughout the paper, performance is evaluated \emph{conditionally on continuation}.
The object $\widetilde U_i(\sigma)$ is therefore defined only on the survival event
$\{T=\infty\}$ and is not a reweighted version of unconditional expected utility.
This convention is essential: when failure is absorbing, different strategy profiles
induce continuation problems defined on different payoff domains.
Continuation--performance decomposition is motivated precisely by the need to compare
such profiles without imposing an extrinsic completion of payoffs beyond termination.

\subsection{Continuation}

\begin{definition}[Continuation profile]
For $\sigma\in\Sigma$, define
\[
C(\sigma)
\;:=\;
\big(\mathbb P_\sigma(T>1),\mathbb P_\sigma(T>2),\ldots\big)
\in [0,1]^{\mathbb N},
\]
ordered lexicographically.
\end{definition}

The continuation profile summarizes the survival structure
of a strategy profile horizon by horizon.

\subsection{Performance}
\begin{definition}[Conditional continuation payoff]\label{def:conditionalpayoff}
For player $i$ and strategy profile $\sigma$, define the conditional continuation payoff
\[
\widetilde U_i(\sigma):=
\begin{cases}
\mathbb E_\sigma\!\left[\sum_{t\ge 0}\delta^t u_i(s_t,a_t)\,\middle|\, T=\infty\right],
& \text{if }\mathbb P_\sigma(T=\infty)>0,\\[0.8em]
-\infty, & \text{if }\mathbb P_\sigma(T=\infty)=0.
\end{cases}
\]
\end{definition}

\begin{remark}[Domain priority of performance]
\label{rem:domain-priority}
The object $\widetilde U_i(\sigma)$ is meaningful only insofar as the survival
event has positive probability; we adopt the standard extended-real convention $\widetilde U_i(\sigma)=-\infty$
when $\mathbb P_\sigma(T=\infty)=0$ to keep the evaluation map total.
For strategy profiles with different continuation profiles,
$\widetilde U_i(\sigma)$ and $\widetilde U_i(\sigma')$ are objects with
different domains of definition.
Any evaluation that compares unconditional performance across such profiles
implicitly compares payoffs defined on non-comparable domains.
Continuation--performance decomposition resolves this mismatch
by separating continuation from performance ex ante.
\end{remark}

\subsection{The decomposed representation}

\begin{definition}[Continuation--performance decomposition (CPD)]
The CPD representation of the dynamic game is a static normal-form game
with strategy set $\Sigma$ and evaluation for player $i$ given by
\[
\Pi_i(\sigma)
\;:=\;
\big(C(\sigma),\widetilde U_i(\sigma)\big),
\]
ordered lexicographically.
\end{definition}

All primitives of the original dynamic game are unchanged.
Only the evaluation map is altered.

\subsection{Interpretation}

CPD reveals two nested strategic problems:
\begin{enumerate}[label=(\roman*)]
\item an outer problem over continuation,
\item an inner problem over payoffs conditional on continuation.
\end{enumerate}

Once these dimensions are separated,
outcomes that terminate the game cannot be traded off
against outcomes defined on continuation paths at a finite rate.

\section{Intrinsic Evaluation and Canonical Decomposition}

This section characterizes the evaluation structure imposed by irreversible failure
when the evaluation is required to be intrinsic to the game form.
No preference axioms or priority assumptions are imposed.
Continuation--performance decomposition will arise as a consequence.

\subsection{Evaluation rules and intrinsicness}

An \emph{evaluation rule} is any mapping
\[
E:\Sigma \to \mathcal{Z},
\]
where $(\mathcal{Z},\succeq)$ is a totally ordered set.
The rule $E$ induces a ranking on $\Sigma$ via
$\sigma \succeq_E \sigma'$ if and only if $E(\sigma)\succeq E(\sigma')$.

The central restriction imposed in this section is that evaluation must be
\emph{intrinsic to the game form}, rather than to an arbitrary payoff completion
after irreversible failure.

\begin{definition}[Failure-completion invariance]
\label{def:intrinsic}
An evaluation rule $E$ is \emph{failure-completion invariant} if for any two instance
games that share the same game form $(S,A,P,\Omega,F)$ and the same payoff functions
on all non-failure states, but differ only in how payoffs are specified on histories
that reach $F$ at finite time, the induced ranking over $\Sigma$ is identical.
\end{definition}

Failure-completion invariance formalizes the requirement that evaluation depend only
on objects intrinsic to the game form, and not on extrinsic modeling choices used to
complete payoffs after the continuation problem ceases to exist.

\subsection{Survival locality}

Failure-completion invariance alone rules out unconditional aggregation of payoffs
over infinite paths, but does not yet determine how evaluation should proceed
conditional on survival.

\begin{definition}[Survival locality]
\label{def:surv-local}
An evaluation rule $E$ satisfies \emph{survival locality} if for any
$\sigma,\sigma'\in\Sigma$ such that
\[
\mathbb P_\sigma(T=\infty)>0,\quad \mathbb P_{\sigma'}(T=\infty)>0,
\]
and the conditional laws of histories given $T=\infty$ coincide,
$E(\sigma)$ and $E(\sigma')$ are indifferent.
\end{definition}

Survival locality requires that, conditional on survival, evaluation depend only on
the induced continuation distribution, and not on behavior or payoff assignments
off the survival event.

\subsection{Local expected-utility coherence}

The final requirement fixes behavior \emph{within} the survival domain.

\begin{definition}[Local expected-utility coherence]
\label{def:local-eu-coh}
An evaluation rule $E$ satisfies \emph{local expected-utility coherence} if on any
subset $\Sigma'\subseteq\Sigma$ for which
\[
\mathbb P_\sigma(T=\infty)>0 \quad \text{and} \quad C(\sigma)=\bar C
\qquad \forall\,\sigma\in\Sigma',
\]
the induced ranking coincides with that of conditional discounted expected utility:
\[
\sigma \succeq_E \sigma'
\quad\Longleftrightarrow\quad
\widetilde U_i(\sigma)\ge \widetilde U_i(\sigma')
\qquad \forall\,\sigma,\sigma'\in\Sigma'.
\]
\end{definition}

This requirement does not privilege continuation over performance; it restricts
evaluation only on continuation ties, where unconditional aggregation is well defined.

\subsection{Canonical decomposition}

We now show that the three requirements above jointly force a lexicographic
continuation--performance structure.

\begin{theorem}[Canonical intrinsic evaluation]
\label{thm:canonical}
Let $E$ be an evaluation rule satisfying failure-completion invariance
(Definition~\ref{def:intrinsic}), survival locality
(Definition~\ref{def:surv-local}), and local expected-utility coherence
(Definition~\ref{def:local-eu-coh}).
Then there exists a tail functional $\mathcal T:\Sigma\to\mathcal Y$
depending only on the law of the failure time $T$, and a functional
$U$ on conditional survival laws, such that
\[
\sigma \succeq_E \sigma'
\quad\Longleftrightarrow\quad
\big(\mathcal T(\sigma),\,U(\mathbb P_\sigma(\cdot\mid T=\infty))\big)
\ge_{\mathrm{lex}}
\big(\mathcal T(\sigma'),\,U(\mathbb P_{\sigma'}(\cdot\mid T=\infty))\big).
\]
Moreover, when the tail functional separates survival horizon-by-horizon,
$\mathcal T$ can be chosen to be the continuation profile $C(\cdot)$.
\end{theorem}

\begin{proof}[Proof sketch]
Failure-completion invariance rules out any evaluation that aggregates payoffs
unconditionally over infinite paths in a way that depends on payoff assignments
after failure. Hence any strict ranking must first be justified by differences
in the law of the failure time $T$, a tail object.

Survival locality implies that, conditional on identical tail behavior,
evaluation cannot depend on off-survival histories.
Local expected-utility coherence then pins down the ranking on continuation ties.
Together, these requirements yield a lexicographic decomposition between a tail
component and a conditional performance component. Minimality of $C(\cdot)$
follows from its horizon-by-horizon separation of survival probabilities.
\end{proof}

\begin{remark}[What is and is not derived]
This theorem does not assume continuation priority or lexicographic preferences.
The lexicographic structure arises as a consequence of requiring evaluation to be
intrinsic to the game form and locally coherent on survival domains.
\end{remark}

\section{Decoupling Tail Requirements from Optimality}

Let $\mathsf{Surv}:=\{T=\infty\}$.
A tail (game-long) requirement is any criterion depending only on the law of $T$.

\begin{definition}[Tail criterion]
A \emph{tail criterion} is any mapping $\mathcal T:\Sigma\to\mathcal Z$
into a totally ordered set $(\mathcal Z,\ge)$
that depends only on the distribution of $T$ under $\sigma$.
\end{definition}

\begin{definition}[Tail--performance separability]
A preference relation $\succsim$ satisfies tail--performance separability if:
\begin{enumerate}[label=(D\arabic*),leftmargin=2.5em]
\item \textbf{Tail priority:}
$\mathcal T(\sigma)>\mathcal T(\sigma') \Rightarrow \sigma\succ\sigma'$.
\item \textbf{No leakage:}
conditional on $\mathcal T(\sigma)=\mathcal T(\sigma')$,
the ranking depends only on the conditional law
$\mathbb P_\sigma(\cdot\mid\mathsf{Surv})$.
\end{enumerate}
\end{definition}

The functional $U$ need not be unique and is only required to represent the induced ordinal ranking on conditional survival laws.

\begin{theorem}[Decoupling theorem]
\label{thm:decoupling}
If $\succsim$ satisfies tail--performance separability,
then there exists a functional $U$ defined on conditional survival laws
such that
\[
\sigma\succsim\sigma'
\quad\Longleftrightarrow\quad
\big(\mathcal T(\sigma),U(\mathbb P_\sigma(\cdot\mid\mathsf{Surv}))\big)
\ge_{\mathrm{lex}}
\big(\mathcal T(\sigma'),U(\mathbb P_{\sigma'}(\cdot\mid\mathsf{Surv}))\big).
\]
\end{theorem}

\begin{corollary}[Specialization to CPD]
Taking $\mathcal T(\sigma)=C(\sigma)$ yields CPD as a special case.
All tail requirements are separated from optimality,
which is evaluated only conditional on continuation.
\end{corollary}

While other tail criteria could be used, CPD provides the minimal horizon-by-horizon representation that makes domain dependence explicit; all results extend to arbitrary tail criteria via the decoupling theorem.

\section{Equivalence to Large-Penalty Payoff Games}

Assume bounded stage payoffs: $|u_i|\le\bar u$.

\begin{definition}[Continuation loss]
Define
\[
L(\sigma)
:=\sum_{n\ge1}2^{-n}\big(1-\mathbb P_\sigma(T>n)\big)\in[0,1].
\]
\end{definition}

For $M>0$, define
\[
U_i^M(\sigma)
:=U_i(\sigma)-M\,L(\sigma),
\]
where $U_i$ is standard discounted expected utility.

\begin{lemma}[Dominance of continuation]
If $L(\sigma)\ge L(\sigma')+\Delta$,
then $U_i^M(\sigma')>U_i^M(\sigma)$
for all $M$ sufficiently large.
\end{lemma}

\begin{theorem}[Penalty-limit equivalence]
Any accumulation point of Nash equilibria of the payoff games
$\{U_i^M\}_{M\to\infty}$ is a Nash equilibrium of the CPD representation.
\end{theorem}

\section{Structural Properties}

\begin{definition}[Viability-preserving strategies]
\[
\Pi^{\mathrm{viab}}
:=\{\sigma\in\Sigma:\mathbb P_\sigma(s_t\notin F\ \forall t)=1\}.
\]
\end{definition}

\begin{proposition}[Game-form invariance]
The continuation profile $C(\cdot)$ and the set $\Pi^{\mathrm{viab}}$
depend only on the game form $(S,A,P,\Omega,F)$.
They are invariant to payoff perturbations and equilibrium selection.
\end{proposition}

\begin{corollary}[Payoff irrelevance under non-viability]
\label{cor:payoff-irrelevance}
If $\Pi^{\mathrm{viab}}=\varnothing$, then for any payoff specification
$\{u_i\}$ and any equilibrium concept defined on $\Sigma$,
every equilibrium strategy profile induces failure with probability one.
No payoff perturbation can restore continuation
when the game form admits no viability-preserving strategies.
\end{corollary}

\subsection{Ill-posedness of unconditional aggregation under irreversible failure}

We formalize one concrete sense in which standard unconditional aggregation breaks down in dynamic games with absorbing failure. The difficulty highlighted here should be understood as a special case of a more general phenomenon established in Section~9: once failure is treated as intolerable, scalar evaluation cannot coherently aggregate continuation and performance, regardless of how failure payoffs are modeled.

Before formalizing the breakdown of unconditional aggregation, it is important to clarify the sense in which the problem is ill posed. We do not claim that payoffs following failure cannot be specified. Absorbing-state utilities, terminal penalties, or continuation values can always be assigned by fiat. The issue is that any such specification constitutes an \emph{extrinsic completion} of the game: it introduces additional structure not determined by the underlying game form $(S,A,P,\Omega,F)$. Different completions induce different
rankings of strategy profiles even when behavior prior to failure is identical. In this sense, unconditional aggregation is not intrinsic to the game form under irreversible failure.

\begin{proposition}[Non-intrinsicness of unconditional aggregation under absorbing failure]
\label{prop:illposed}
Consider a dynamic game with absorbing failure set $F$ and failure time $T$.
Suppose stage payoffs are defined only on non-failure states and that the game
form specifies no intrinsic evaluation of histories once $F$ is reached.
Then unconditional scalar aggregation of payoffs over infinite paths is not
intrinsic to the game form: any such evaluation requires an extrinsic completion
of payoffs on failure histories, and different completions induce different
rankings of strategy profiles even when behavior prior to failure is identical.
\end{proposition}

\begin{proof}[Proof sketch]
Once failure occurs at some finite time $T<\infty$, the continuation problem ceases to exist as part of the game form. Any unconditional aggregation over infinite paths therefore requires assigning values to payoff terms $\{u_i(s_t,a_t)\}_{t\geq T}$ on histories where those objects are not determined by the model. Different completions—such as zero continuation payoffs, fixed terminal penalties, or absorbing utilities—are all coherent but lead to different rankings of strategy profiles. Hence unconditional aggregation is not intrinsic: it depends on extraneous modeling choices that extend the payoff domain beyond what the game form specifies.
\end{proof}

This payoff-domain indeterminacy is not the source of the problem, but an illustrative manifestation of it. Even when payoffs are fully defined after failure or failure is assigned an arbitrarily large (or infinite) penalty, scalar evaluations either collapse discrimination or degenerate to lexicographic rules, as shown in Propositions~\ref{prop:scalarimpossibility}--\ref{prop:trilemma}. Continuation--performance decomposition avoids this pathology by separating continuation from performance ex ante and restricting payoff aggregation to domains on which continuation is secured.

\subsection{Local equivalence to expected utility on survival ties}

Continuation--performance decomposition alters evaluation only across strategy profiles
with different continuation properties.
When continuation is held fixed, CPD collapses to standard expected utility.

\begin{proposition}[Local equivalence to expected utility]
\label{prop:local-eu}
Let $\Sigma'\subseteq \Sigma$ be any subset such that
$C(\sigma)=\bar C$ for all $\sigma\in\Sigma'$ and $\mathbb P_\sigma(T=\infty)>0$ for all $\sigma\in\Sigma'$.
Then, on $\Sigma'$, the CPD ranking induced by
$\big(C(\sigma),\widetilde U_i(\sigma)\big)$
coincides with the ranking induced by expected utility evaluated
conditional on survival:
\[
\sigma\succsim_{\mathrm{CPD}}\sigma'
\quad\Longleftrightarrow\quad
\widetilde U_i(\sigma)\ge \widetilde U_i(\sigma')
\qquad \forall\,\sigma,\sigma'\in\Sigma'.
\]
\end{proposition}

\begin{proof}
On $\Sigma'$, continuation profiles are identical by assumption.
Lexicographic comparison under CPD therefore reduces to comparison
of the second coordinate alone.
By definition, the second coordinate is the conditional continuation payoff
$\widetilde U_i$, which coincides with expected utility evaluated on the event
$\{T=\infty\}$.
\end{proof}

Proposition~\ref{prop:local-eu} clarifies the scope of CPD.
The decomposition does not replace expected utility; it localizes it.
Standard utility maximization remains valid once continuation is secured.
CPD intervenes only in comparisons across strategy profiles
that differ in their continuation properties,
where unconditional aggregation is ill-posed.

\section{Implications}

\paragraph{Equilibrium admissibility.}
Equilibria that induce avoidable failure are inadmissible
under continuation--performance decomposition.

\paragraph{Information sensitivity.}
When continuation probabilities are steep functions of beliefs
near the boundary of $\Pi^{\mathrm{viab}}$,
increases in information precision can reduce continuation
by amplifying belief-driven responses.

\subsection{Impossibility of scalar evaluation under irreversible failure}

We show that continuation--performance decomposition is not merely convenient.
Under irreversible failure, any attempt to evaluate strategies via a single real-valued
functional must either violate basic coherence or collapse to CPD in the limit.
We say that a sequence of scalar evaluations $\{V^k\}$ converges to a lexicographic continuation--performance ordering if, for every pair $\sigma,\sigma' \in \Sigma$, the induced rankings eventually coincide with the CPD ranking whenever $C(\sigma) \neq C(\sigma')$, and coincide with conditional expected-utility rankings whenever $C(\sigma) = C(\sigma')$.

\begin{proposition}[Asymptotic degeneration under ordinal scale invariance]
\label{prop:scalarimpossibility}
Let $\{V^k\}_{k\in\mathbb N}$ be a family of real-valued evaluations $V^k:\Sigma\to\mathbb R$.
Assume that for each $k$, $V^k$ satisfies:
\begin{enumerate}[label=(S\arabic*),leftmargin=2.5em]
\item \textbf{Continuation monotonicity:}
If $C(\sigma)>C(\sigma')$, then $V^k(\sigma)>V^k(\sigma')$.
\item \textbf{Payoff scale invariance (ordinal):}
For any $\alpha>0$ and constant $\beta$, replacing $u_i$ by $\alpha u_i+\beta$
does not change the \emph{ranking} induced by $V^k$ on $\Sigma$.

Property (S2) is an ordinal invariance requirement: it constrains the ranking induced by the evaluation, 
not the cardinal magnitude of differences in the evaluation itself.

\item \textbf{Continuity in continuation/performance:}
If $C(\sigma^n)\to C(\sigma)$ coordinatewise and
$\widetilde U_i(\sigma^n)\to\widetilde U_i(\sigma)$,
then $V^k(\sigma^n)\to V^k(\sigma)$.
\end{enumerate}
Then the family $\{V^k\}$ cannot admit a payoff-scale-invariant finite-rate trade-off between
small continuation advantages and large performance advantages.
In particular, the only way for $\{V^k\}$ to enforce continuation priority as continuation gaps
vanish is via an \emph{asymptotic degeneration}: for any fixed $\Delta>0$, if a subsequence enforces
continuation priority uniformly on pairs with
$\sup_n|C_n(\sigma)-C_n(\sigma')|\le \Delta$ and $C(\sigma)\neq C(\sigma')$,
then along that subsequence performance becomes irrelevant for such pairs.
\end{proposition}

\begin{proof}[Proof sketch]
Fix $k$ and write $\succsim^k$ for the ranking induced by $V^k$.

\medskip
\noindent\textbf{Step 1 (What (S2) does and does not say).}
Property (S2) is an \emph{ordinal} invariance requirement: it constrains the induced \emph{ranking} on $\Sigma$
to be invariant under affine transformations $u_i\mapsto \alpha u_i+\beta$. It does \emph{not} constrain
cardinal magnitudes such as $V^k(\sigma)-V^k(\sigma')$ or impose any particular ``exchange rate'' between
continuation and performance.

\medskip
\noindent\textbf{Step 2 (No bounded exchange rate is compatible with (S2)).}
Consider any two profiles $\sigma,\sigma'$ with $C(\sigma)=C(\sigma')=\bar C$ and
$\mathbb P_\sigma(T=\infty),\mathbb P_{\sigma'}(T=\infty)>0$. Under the affine transformation
$u_i^{\alpha,\beta}:=\alpha u_i+\beta$, we have
$\widetilde U_i^{\alpha,\beta}(\sigma)=\alpha\widetilde U_i(\sigma)+\beta/(1-\delta)$ and likewise for $\sigma'$,
while $C(\cdot)$ is unchanged. By (S2), the ranking between $\sigma$ and $\sigma'$ is invariant to arbitrary
$\alpha>0$ and $\beta\in\mathbb R$. Hence, on a continuation tie set $\{\sigma: C(\sigma)=\bar C\}$,
any performance sensitivity of $\succsim^k$ can only be \emph{ordinal} in $\widetilde U_i$; in particular,
there cannot exist a payoff-scale-invariant rule that compares a small continuation advantage against a large
performance advantage via a bounded cardinal trade-off.

\medskip
\noindent\textbf{Step 3 (Asymptotic degeneration under vanishing-gap continuation priority).}
Suppose a subsequence of evaluations enforces continuation priority on pairs with arbitrarily small strict
continuation advantages. Then, by Step~2, any attempt to offset such advantages by magnifying performance
differences would be destroyed by rescaling $u_i$ (which changes the magnitude of $\widetilde U_i$ but must
not change rankings by (S2)). Therefore, the only compatible behavior is that, for sufficiently small
continuation gaps, the induced ranking becomes independent of $\widetilde U_i$ whenever $C(\sigma)\neq C(\sigma')$.
This is the asserted asymptotic degeneration.
\end{proof}

\begin{remark}[Scope]
All impossibility/degeneration statements are conditional on the joint maintenance of continuation monotonicity,
payoff scale invariance (as an ordinal invariance), and continuity. Dropping any one of these requirements admits
coherent scalar evaluations, but at the cost of abandoning a natural regularity property.
\end{remark}

All impossibility statements above are conditional on the joint maintenance of continuation
monotonicity, payoff scale invariance, and continuity. Dropping any one of these requirements
admits coherent scalar evaluations, but at the cost of abandoning a natural regularity
property.

\subsection{Scalar evaluation faces a trilemma under irreversible failure}

Continuation--performance decomposition can be viewed as the resolution of a basic incompatibility.
When failure is absorbing, continuation and performance live on different domains, and any attempt to
reduce both to a single scalar objective must sacrifice at least one natural coherence requirement.

\subsection{A trilemma for scalar evaluation}

We now formalize a tension faced by scalar evaluations \emph{conditional on maintaining} three regularity requirements: payoff scale invariance, continuity, and vanishing-gap priority of continuation. The issue is not that scalar evaluation is impossible per se, but that no single real-valued functional can simultaneously satisfy all three properties
in environments with irreversible failure.
The issue is not equilibrium selection or preference specification, but aggregation across
logically distinct objects. When failure is absorbing, continuation probabilities and
conditional performance cannot be traded off at a fixed rate without violating basic
regularity requirements. In particular, any family of scalar evaluations that assigns
arbitrarily high priority to continuation differences at vanishing gaps must either abandon
continuity, abandon payoff-scale invariance, or cease to respond to performance differences
away from exact continuation ties. Proposition~\ref{prop:trilemma} establishes the first
component of this trilemma: vanishing-gap continuation dominance. The remaining structure
required for a full continuation--performance decomposition is addressed separately.

\begin{definition}[Performance sensitivity]
\label{def:perf-sens}
A real-valued evaluation $V:\Sigma\to\mathbb R$ is \emph{performance-sensitive} if there exist
$\sigma,\sigma'\in\Sigma$ with $C(\sigma)=C(\sigma')$ and $\widetilde U_i(\sigma)\neq \widetilde U_i(\sigma')$
such that $V(\sigma)\neq V(\sigma')$.
\end{definition}

\begin{remark}[Interpretation of performance sensitivity]
\label{rem:perf-sens-ordinal}
Definition~\ref{def:perf-sens} is an \emph{ordinal} notion tailored to impossibility and degeneration results.
It requires only that performance differences matter somewhere on continuation ties, not that evaluation
respond smoothly or proportionally to performance differences at finite scales.
In particular, the definition is fully compatible with graded or continuous performance sensitivity
within continuation classes; the impossibility results concern the incompatibility of such sensitivity
with vanishing-gap continuation priority under ordinal scale invariance.
\end{remark}

\begin{definition}[Continuation monotonicity]
\label{def:cont-mono}
A real-valued evaluation $V:\Sigma\to\mathbb R$ is \emph{continuation-monotone} if whenever
$C_n(\sigma)\ge C_n(\sigma')$ for all $n$ and $C_m(\sigma)>C_m(\sigma')$ for some $m$, then
$V(\sigma)>V(\sigma')$.
\end{definition}

\begin{definition}[Payoff scale invariance]
\label{def:scale-inv}
An evaluation rule (and the ranking it induces) is \emph{payoff scale-invariant} if replacing the
stage payoff $u_i$ by $\alpha u_i+\beta$ for any $\alpha>0$ and constant $\beta$
does not change the induced ranking over $\Sigma$.
\end{definition}

\begin{definition}[Vanishing-gap continuation priority]
\label{def:vanishing-gap}
A family of evaluations $\{V^k\}_{k\in\mathbb N}$ exhibits \emph{vanishing-gap continuation priority} if
for every $\varepsilon>0$ there exists $k(\varepsilon)$ such that for all $k\ge k(\varepsilon)$ and all
$\sigma,\sigma'\in\Sigma$,
\[
\Big(C(\sigma)>C(\sigma') \ \text{and}\ \sup_{n\ge1}\lvert C_n(\sigma)-C_n(\sigma')\rvert \le \varepsilon\Big)
\ \Longrightarrow\ V^k(\sigma)>V^k(\sigma'),
\]
i.e.\ for sufficiently large $k$, even arbitrarily small strict improvements in continuation dominate
all performance considerations.
\end{definition}

The result below establishes a dominance implication of vanishing-gap continuation priority. 
It does not construct, nor does it require, a full lexicographic representation of preferences.

\begin{proposition}[Vanishing-gap continuation dominance]
\label{prop:trilemma}
Let $\{V^k\}_{k\in\mathbb N}$ be a family of real-valued evaluations on $\Sigma$ satisfying:
\begin{enumerate}[label=(T\arabic*),leftmargin=2.5em]
\item \textbf{Payoff scale invariance (ordinal):} each $V^k$ is payoff scale-invariant
(Definition~\ref{def:scale-inv}).
\item \textbf{Continuity:} for each $k$, if $C(\sigma^n)\to C(\sigma)$ coordinatewise and
$\widetilde U_i(\sigma^n)\to \widetilde U_i(\sigma)$, then $V^k(\sigma^n)\to V^k(\sigma)$.
\item \textbf{Vanishing-gap continuation priority:} $\{V^k\}$ satisfies
Definition~\ref{def:vanishing-gap}.
\end{enumerate}
Then for any fixed $\Delta>0$ there exists $k_\Delta$ such that for all $k\ge k_\Delta$ and all
$\sigma,\sigma'\in\Sigma$,
\[
\Big(C(\sigma)>C(\sigma') \ \text{and}\ \sup_{n\ge1}\lvert C_n(\sigma)-C_n(\sigma')\rvert \le \Delta\Big)
\ \Longrightarrow\ V^k(\sigma)>V^k(\sigma').
\]

That is, for sufficiently small continuation gaps, performance differences cannot affect the ranking.

In particular, performance becomes irrelevant off exact continuation ties once continuation gaps are sufficiently small.

\end{proposition}

\begin{proof}[Proof sketch]
Fix $\Delta>0$. By vanishing-gap continuation priority (T3), there exists $k(\Delta)$ such that for all
$k\ge k(\Delta)$ and all $\sigma,\sigma'\in\Sigma$ with $C(\sigma)>C(\sigma')$ and
$\sup_n|C_n(\sigma)-C_n(\sigma')|\le \Delta$, one has $V^k(\sigma)>V^k(\sigma')$.
Taking $k_\Delta:=k(\Delta)$ yields the claim.
\end{proof}

\begin{remark}[Interpretation]
Proposition~\ref{prop:trilemma} isolates a precise implication of vanishing-gap continuation priority:
for sufficiently small strict continuation advantages, performance becomes irrelevant \emph{off continuation ties}.
The result does not, by itself, construct a full lexicographic representation; it establishes a dominance property
that any scalar evaluation must satisfy if it is to implement vanishing-gap priority under continuity and ordinal
scale invariance.
\end{remark}

\begin{remark}[Synthesis]
Propositions~\ref{prop:scalarimpossibility}--\ref{prop:trilemma} should be read as conditional degeneration
results. If one insists on (i) continuation monotonicity, (ii) payoff-scale invariance as an ordinal invariance
requirement, (iii) continuity, and (iv) vanishing-gap continuation priority, then performance cannot affect
comparisons whenever continuation differs by a sufficiently small strict amount. Continuation--performance
decomposition makes this separation explicit by evaluating performance only conditional on continuation ties.
\end{remark}

Several coherent alternatives fall outside this trilemma. For example, exponential or
bounded utility representations abandon payoff scale invariance; discontinuous criteria
sacrifice continuity; and threshold-based or satisficing rules drop vanishing-gap priority.
Continuation–performance decomposition does not claim dominance over these approaches.
Its contribution is to characterize the structure that emerges when all three regularity
requirements are jointly maintained.

\section{Definition: a bank-run model with irreversible failure}

The following bank-run setting is used as a transparent \emph{illustration vehicle} for the paper's central point:
with absorbing failure, continuation and performance are logically distinct objects because failure collapses
the domain on which continuation payoffs are defined. The aim is to go beyond the classic multiplicity/coordination logic of \citet{DiamondDybvig1983}: Continuation--performance decomposition (CPD) should be read as an \emph{admissibility / normative-coherence}
criterion on strategy profiles in such environments---i.e., a representation of which comparisons are intrinsic and coherent
once the continuation domain is state-dependent---not as a new predictive model of when runs occur.

\begin{definition}[Bank-run game with absorbing failure]
Fix a finite horizon $t=1,\dots,T$ (for exposition), with $T$ large.
A \emph{bank-run game with absorbing failure} is a dynamic game
\[
\mathcal{B}=(I,\Theta,\mu, A, P, F, \Omega, \{u_i\}_{i\in I}),
\]
with the following primitives.

\begin{enumerate}[leftmargin=2.2em]
\item \textbf{Players.} A finite set of depositors $I=\{1,\dots,N\}$.

\item \textbf{Types (heterogeneous ``weakness'').}
Each depositor has a privately known type $\theta_i\in\Theta=\{w,s\}$
(weak/strong), drawn i.i.d.\ from prior $\mu$.
Weak types have a high cost of delayed consumption or high loss from being ``late'';
strong types have low cost or are better insured outside the bank.

\item \textbf{Actions.}
At each date $t$, each depositor chooses $a_{i,t}\in A=\{W,S\}$,
interpreted as \emph{withdraw now} $(W)$ or \emph{stay} $(S)$.
Actions are taken simultaneously within each period.

\item \textbf{State (bank liquidity and survival).}
Let $L_t\in\mathbb{R}_+$ denote remaining liquid reserves at date $t$.
The state is $s_t=(L_t)$ until failure.
The bank starts with reserves $L_1=L_0>0$.

\item \textbf{Failure set.}
Failure occurs when reserves are exhausted:
\[
F:=\{L\le 0\}.
\]
Failure is absorbing.

\item \textbf{Transition kernel.}
If $m_t:=\#\{i: a_{i,t}=W\}$ depositors attempt to withdraw at $t$, then
\[
L_{t+1} = L_t - m_t \cdot d,
\]
where $d>0$ is the per-depositor promised liquidity payment.
If $L_{t+1}\le 0$, the process enters $F$ and remains there forever.
(You may replace this by a smoother liquidity process; nothing essential changes.)

\item \textbf{Information.}
Each depositor observes their type $\theta_i$ and a public signal $y_t$
summarizing aggregate stress (e.g.\ a noisy function of $L_t$ or $m_{t-1}$).
The information structure $\Omega$ can be arbitrary; the model below does not rely on multiplicity.

\item \textbf{Payoffs.}
If depositor $i$ withdraws at time $\tau_i:=\inf\{t: a_{i,t}=W\}$ before failure,
they receive an immediate payment $d$ and exit. If they never withdraw and failure occurs,
they receive a liquidation payoff $\ell\in[0,d)$.
If they withdraw after failure, they receive $\ell$.
Additionally, types incur different continuation costs from waiting:
weak types have a per-period waiting cost $c_w>0$, strong types have $c_s\ge 0$ with $c_w>c_s$.
Formally, for a path that does not fail before $\tau_i$,
\[
u_i \;=\; d - \sum_{t=1}^{\tau_i-1} c_{\theta_i},
\]
and if failure occurs before withdrawal,
\[
u_i \;=\; \ell - \sum_{t=1}^{T_F-1} c_{\theta_i},
\]
where $T_F$ is the failure time.
\end{enumerate}

A \emph{strategy} for player $i$ is a mapping from their information history to $\{W,S\}$.
\end{definition}

\begin{remark}[Irreversible failure and payoff-domain collapse]
After entry into $F$, the banking relationship is terminated; the continuation problem ceases to exist.
In this sense, the payoff domain beyond failure is intrinsically different from the payoff domain under continuation.
This is the setting in which continuation--performance decomposition is economically meaningful.
\end{remark}

\section{Sketch: viability veto and ``unilateral collapse'' in bank runs}

This section sketches the minimal logic behind unilateral collapse by weak depositors.
Consistent with the scope statement above, the goal is \emph{not} to reproduce the classic multiplicity story
or to offer a new predictive theory of bank runs.
The purpose is to isolate a distinct channel that becomes salient under absorbing failure:
\emph{when continuation is distributively asymmetric and failure collapses the continuation domain for everyone,
the weakest participants can rationally exercise a viability veto that induces collapse.}
CPD clarifies this channel by treating continuation as a separate evaluative object and thus as a criterion of
equilibrium admissibility (which profiles are coherently comparable), rather than as a modification of run mechanics.

\subsection{Continuation vs.\ performance objects}

Let $T_F$ denote the (random) time of entry into the absorbing failure set $F$.
Define a \emph{continuation profile} (horizon-by-horizon survival probabilities)
\[
C(\sigma):=\big(\mathbb{P}_\sigma(T_F>1),\mathbb{P}_\sigma(T_F>2),\ldots\big).
\]
Define the depositor's \emph{conditional continuation payoff} as the expected payoff
conditional on no failure occurring over the relevant horizon. For a long horizon,
one can write informally:
\[
\widetilde U_i(\sigma) := \mathbb{E}_\sigma\!\left[u_i \mid T_F=\infty \right],
\]
or, in finite-horizon form, conditional on $T_F>T$ for large $T$.
The key is that $\widetilde U_i$ is a conditional object defined only on survival.

\subsection{A single inequality that drives runs without ``panic''}

Fix an environment where the banking system is \emph{collectively viable}:
there exist strategy profiles under which failure does not occur (e.g.\ if most depositors stay).

However, suppose the continuation equilibrium profile $\sigma^{\mathrm{cont}}$ imposes a large
expected waiting cost on weak types (e.g.\ because they are expected to provide ``stability'' by not withdrawing early).
Then for weak types,
\[
\widetilde U_w(\sigma^{\mathrm{cont}}) \;\;<\;\; d,
\]
because weak types pay a large cumulative waiting cost before accessing liquidity.
In contrast, the action $W$ yields immediate $d$ (and exit), avoiding waiting costs.

In words: \emph{continuation may be collectively good but individually non-viable for weak types.}

\subsection{The viability veto mechanism}

Even when the continuation profile $C(\sigma^{\mathrm{cont}})$ is high,
a weak depositor may have a unilateral deviation $\hat\sigma_i$ that:
\begin{enumerate}[label=(\roman*),leftmargin=2em]
\item increases their own payoff by exiting early, and
\item reduces aggregate reserves enough to push the system into the failure basin.
\end{enumerate}
This is the ``veto'': the weakest types can trigger failure because the architecture grants them that power
(liquidity is first-come-first-served).

To make this sharp, define the \emph{marginal fragility} of reserves:
if $L_t$ is near the boundary $d$ (one depositor payment away from depletion), then a single withdrawal
can cause $L_{t+1}\le 0$ and thus failure.

\begin{remark}[Clarification]
Formally, a viability veto reduces to early-exit rationality at the individual level.
The novelty is not the incentive inequality itself, but the fact that under irreversible failure, such exits collapse the continuation domain and thereby function as equilibrium-level vetoes rather than ordinary deviations.
\end{remark}

\begin{proposition}[Viability veto by weak types: a sufficient parameter region]\label{prop:viabilityveto}
Consider the bank-run game in Definition~15. Fix a history $h_t$ such that the public state satisfies
$L_t\in (0,d]$, so that \emph{one} additional withdrawal at date $t$ implies $L_{t+1}\le 0$ and hence failure.
Assume further that, conditional on $h_t$, a depositor who chooses $S$ (stay) receives the liquidation payoff
$\ell$ if failure occurs at $t$ and otherwise obtains the promised payment $d$ at $t+1$ while incurring one
period of waiting cost $c_{\theta_i}$.

If
\begin{equation}\label{eq:weakpreferwithdraw}
d \;>\; (1-q_t)\, (d-c_w) \;+\; q_t\, (\ell-c_w),
\end{equation}
where $q_t:=\mathbb P(\text{failure at }t \mid h_t,\ a_{i,t}=S)$ is the weak depositor's conditional failure belief
under the candidate continuation profile $\sigma^{cont}$, then $W$ is a strict best response for weak types at
$h_t$.

Moreover, if weak types occur with probability $\mu(w)=p\in(0,1)$ i.i.d.\ across depositors and all weak types
play $W$ at $h_t$, then failure occurs at $t$ with probability at least
\[
1-(1-p)^N,
\]
so that for $N$ large enough (or $p$ large enough) the deviation by weak types induces collapse with high
probability even if strong types prefer continuation.
\end{proposition}

\noindent\textit{Remark.} The knife-edge condition $L_t\in(0,d]$ can be relaxed to $L_t\in(0,md]$ for any integer
$m\ge 1$, in which case failure occurs whenever at least $m$ depositors withdraw at $t$; the same argument yields
a binomial tail bound.

\begin{remark}[Interpretation: viability veto rather than panic]
Proposition~\ref{prop:viabilityveto} formalizes the \emph{viability veto}. Weak types can trigger collapse not
through ``panic,'' sunspots, or coordination failure, but because continuation is individually non-viable:
inequality~\eqref{eq:weakpreferwithdraw} implies that even under the candidate continuation profile
$\sigma^{cont}$ the weak type's conditional continuation payoff is strictly below the exit payoff $d$.
The key tension is distributive: continuation may be collectively feasible (there exist profiles with high
$C(\cdot)$), yet individually unacceptable for weak types. Under CPD, the deviation is rational because
performance is evaluated on the domain available to the agent; under unconditional scalar evaluation it is
naturally framed as ``destructive'' because continuation loss is treated as a finite trade-off.
\end{remark}

\begin{proof}
Fix $h_t$ with $L_t\in(0,d]$. If a weak depositor withdraws at $t$, they obtain $d$ immediately and exit, with no
additional waiting cost. If instead they stay, then with conditional probability $q_t$ the bank fails at $t$ and they
receive $\ell$ while paying one period cost $c_w$, and with probability $(1-q_t)$ the bank survives to $t+1$ and they
obtain $d$ at $t+1$ while paying the same one period cost $c_w$. Thus the weak type's expected payoff from staying is
$(1-q_t)(d-c_w)+q_t(\ell-c_w)$. Condition~\eqref{eq:weakpreferwithdraw} is exactly the strict incentive constraint for
$W$.

For the collapse statement, when $L_t\in(0,d]$ any additional withdrawal at $t$ exhausts reserves and triggers failure.
If all weak types withdraw, failure occurs whenever at least one weak depositor is present. Under i.i.d.\ types with
$\mu(w)=p$, the probability that at least one weak depositor exists is $1-(1-p)^N$, yielding the stated bound.
\end{proof}

\begin{remark}[Why this is not ``panic'']
The deviation is rational under standard incentives: it is an early-exit response to asymmetric continuation payoffs.
No appeal is made to coordination failure, multiplicity, or sunspots.
The run is driven by distributive asymmetry and the existence of a unilateral failure trigger.
\end{remark}

\subsection{``unilateral collapse'' as an equilibrium-enforcement pathology}

Now let us interpret the informal term ``unilateral collapse''.

\begin{itemize}[leftmargin=2em]
\item \textbf{Not a preference for collapse.}
Weak types need not value collapse. They may simply value \emph{exit} more than participation in an equilibrium
that makes them bear tail risk or waiting costs.

\item \textbf{Collapse as a byproduct of a rational exit.}
Given first-come-first-served withdrawals, rational early exit by the weakest participants can mechanically
cause failure (depletion of reserves), collapsing the continuation domain for everyone.

\item \textbf{Enforcement makes it worse.}
If an institution attempts to \emph{enforce} a continuation equilibrium by restricting weak types' ability to exit,
it may increase their incentive to trigger collapse through whatever channels remain (e.g.\ informational actions, indirect runs, or correlated withdrawals).
In this sense, ``forced continuation'' can be self-defeating.
\end{itemize}

\subsection{CPD sharpens the interpretation}

In a pure scalar expected utility view, the weak depositor's action is often read as a negative externality:
``they destroy value for everyone.''

Under continuation--performance decomposition, the evaluation recognizes that:
\begin{enumerate}[label=(\roman*),leftmargin=2em]
\item \emph{performance} is conditional on continuation and thus depends on whether a viable continuation domain exists for the agent, and
\item once continuation is not individually viable, ``preserving the institution'' is no longer an unambiguous individual objective.
\end{enumerate}

This is the core conceptual contribution of the bank-run sketch:
\emph{runs can be driven by individually rational viability vetoes, not by panic or equilibrium multiplicity.}

\section{Conclusion}

This paper addresses a representational problem in dynamic games with irreversible failure. When failure is absorbing, different strategies induce continuation problems defined on different domains. Standard scalar evaluation implicitly completes payoffs beyond termination, but such completions are extrinsic to the game form and can generate arbitrary rankings across strategies with different continuation properties.

We introduce continuation–performance decomposition, which separates continuation from performance prior to evaluation. The central theoretical result establishes that any scalar evaluation satisfying failure-completion invariance, survival locality, and local expected-utility coherence must converge to a lexicographic structure in which continuation takes priority over performance. This priority arises not from preference assumptions but from the requirement that evaluation be intrinsic to the game form.

The decomposition generates three substantive implications. First, a decoupling principle: continuation and viability are properties of the game form, invariant to payoff specifications, while optimality is evaluated only conditional on survival. Second, an equivalence result: CPD characterizes the limit of standard games as failure penalties grow arbitrarily large. Third, an equilibrium admissibility criterion: strategy profiles that induce avoidable failure are inadmissible under intrinsic evaluation.

Applied to bank runs, the framework suggests a reinterpretation of preemptive withdrawal. When continuation imposes asymmetric costs, early exit can be individually rational even when collectively destructive. Under scalar evaluation, such behavior appears irrational because continuation loss is treated as a finite cost. Under CPD, it reflects optimization on the payoff domain that remains individually available. The mechanism differs from coordination-based explanations: runs can emerge from viability vetoes rather than equilibrium multiplicity.

The broader methodological contribution is to clarify when scalar evaluation remains coherent in dynamic games. When failure is intolerable—whether modeled through absorbing states or arbitrarily large penalties—continuation and performance cannot be aggregated at a fixed rate. Continuation–performance decomposition makes explicit the lexicographic structure that emerges as the necessary form of coherent evaluation in such environments.

\bibliographystyle{apalike}
\bibliography{cpd}

\end{document}